\documentclass{article}
\usepackage[utf8]{inputenc}

    \usepackage{amsmath}
    \usepackage{amsthm}
    \usepackage{amsfonts}
    \usepackage{amssymb}
    \newtheorem{lemma}{Lemma}
    \newtheorem{theorem}{Theorem}
\usepackage{tikz}
\usetikzlibrary{matrix}
\usetikzlibrary{fit}					
\usetikzlibrary{backgrounds}	
\usetikzlibrary{decorations.pathreplacing} 
\newcommand{\sv}{{\mathcal V}}
\newcommand{\se}{{\mathcal E}}
\newcommand{\I}{\mathcal I}
\newcommand{\sca}{\$}
\newcommand{\scb}{{\circ}}
\title{{\sc Consensus Patterns} parameterized by input string length is W[1]-hard.}
\author{Laurent Bulteau}
\begin{document}
\maketitle

We consider the {\sc Consensus Patterns} problem, where, given a set of input strings, one is asked to extract a long-enough pattern which appears (with some errors) in all strings. Formally, the problem is defined as follows:
\begin{quote}
{\sc Consensus Patterns}\\
{\bf Input:} Strings $S_1,\ldots S_n$ of length at most $\ell$, integers $m$ and $d$.\\
{\bf Output:} Length-$m$ string $S$ and integers $(j_1,\ldots, j_n)$ such that $\sum_{i=1}^n \mathrm{Ham}(S,  S_i[j_i..j_i+m-1]) \leq d$
\end{quote}
Where $\mathrm{Ham}()$ denotes the Hamming distance and $S[a..b]$ is the substring of $S$ starting in $a$ and ending in $b$. This problem is one of many variations of the well-studied {\sc Consensus String} problem. It is similar to {\sc Consensus Substring} in that the target string must be close to a substring of each input string (rather than the whole string). However, in the latter problem the distance to \emph{each} input string is bounded, rather than the sum of the distances in our case. 

We look at this problem from the parameterized complexity viewpoint, more precisely for parameter $\ell$. Recall that {\sc Consensus Substring} is FPT for parameter $\ell$ \cite{DBLP:journals/tcs/EvansSW03}. See~\cite{DBLP:journals/eatcs/BulteauHKN14} for an overview of the variants of {\sc Consensus String}, and~\cite{DBLP:conf/mfcs/Schmid15} for recent advances on parameterized aspects of {\sc Consensus Substring} and {\sc Consensus Patterns}. We prove the following result.

\begin{theorem}
{\sc Consensus Patterns}$(\ell)$ is W[1]-hard.
\end{theorem}

By reduction from {\sc Multi-Colored Clique}. We are given a graph $G=(V,E)$, with a partition (coloring) $V=V_1 \cup V_2 \cup \ldots \cup V_k$, such that no edge has both endpoints of the same color. Assume that $|V_h|=n$ for all $h\in [k]$. Write $V_h=\{v_{h,1}, v_{h,2}, \ldots v_{h,n}\}$, i.e. each vertex has an index depending both on its color and its rank within its color. Let $m=|E|$. {\sc Multi-Colored Clique} is W[1]-hard for parameter $k$ \cite{downey2012parameterized}. See Figure~\ref{fig:reduction} for an example of the reduction.

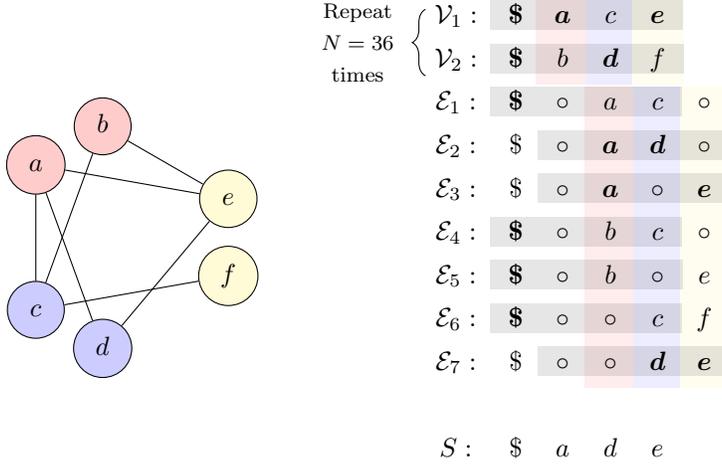
\begin{figure}
\begin{tikzpicture}[baseline={([yshift=-.5ex]current bounding box.center)}]

\node[draw, circle, fill=red!20] (a) at (140:1.5) {\strut $a$};
\node[draw, circle, fill=red!20] (b) at (100:1.5) {\strut $b$};
\node[draw, circle, fill=blue!20] (c) at (220:1.5) {\strut $c$};
\node[draw, circle, fill=blue!20] (d) at (260:1.5) {\strut $d$};
\node[draw, circle, fill=yellow!20] (e) at (20:1.5) {\strut $e$};
\node[draw, circle, fill=yellow!20] (f) at (-20:1.5) {\strut $f$};

\draw (a) --(c);
\draw (a) --(d);
\draw (a) --(e);
\draw (b) --(c);
\draw (b) --(e);
\draw (d) --(e);
\draw (c) --(f);
\begin{scope}[xshift=6cm]
\matrix (m) [matrix of math nodes, row sep=0.8em,
column sep=0.5em, text height=0.4ex, text depth=0ex] {
\sv_1:&\boldsymbol \sca&\boldsymbol a&c&\boldsymbol e\\
\sv_2:&\boldsymbol \sca&b&\boldsymbol d&f\\
\se_1:&\boldsymbol \sca&\scb&a&c&\scb\\
\se_2:&\sca&\scb&\boldsymbol a&\boldsymbol d&\scb\\
\se_3:&\sca&\scb&\boldsymbol a&\scb&\boldsymbol e\\
\se_4:&\boldsymbol \sca&\scb&b&c&\scb\\
\se_5:&\boldsymbol \sca&\scb&b&\scb&e\\
\se_6:&\boldsymbol \sca&\scb&\scb&c&f\\
\se_7:&\sca&\scb&\scb&\boldsymbol d&\boldsymbol e\\
&\ \\
S:&\sca&a&d&e\\
};
\draw [decorate,decoration={brace,amplitude=5pt},xshift=-6pt,yshift=4pt](m-2-1.south west) --(m-1-1.north west)  node [black,midway,xshift=-0.9cm, align=center] {\footnotesize Repeat\\\footnotesize  $N=36$\\\footnotesize  times};
\begin{pgfonlayer}{background}[]
        \node [fill=red!7,    fit=(m-1-3.north west) (m-2-3.south east)] {};
        \node [fill=blue!7,   fit=(m-1-4.north west) (m-2-4.south east)] {};
        \node [fill=yellow!7, fit=(m-1-5.north west) (m-2-5.south east)] {};
        \node [fill=red!7,    fit=(m-3-4.north west) (m-9-4.south east)] {};
        \node [fill=blue!7,   fit=(m-3-5.north west) (m-9-5.south east)] {};
        \node [fill=yellow!7, fit=(m-3-6.north west) (m-9-6.south east)] {};
        
        \node [fill=black,opacity=0.1,    fit=(m-1-2.north west) (m-1-5.east)] {};
        \node [fill=black,opacity=0.1,    fit=(m-2-2.north west) (m-2-5.east)] {};
        \node [fill=black,opacity=0.1,    fit=(m-3-2.north west) (m-3-5.east)] {};
        \node [fill=black,opacity=0.1,    fit=(m-4-3.north west) (m-4-6.east)] {};
        \node [fill=black,opacity=0.1,    fit=(m-5-3.north west) (m-5-6.east)] {};
        \node [fill=black,opacity=0.1,    fit=(m-6-2.north west) (m-6-5.east)] {};
        \node [fill=black,opacity=0.1,    fit=(m-7-2.north west) (m-7-5.east)] {};
        \node [fill=black,opacity=0.1,    fit=(m-8-2.north west) (m-8-5.east)] {};
        \node [fill=black,opacity=0.1,    fit=(m-9-3.north west) (m-9-6.east)] {};

\end{pgfonlayer}
\end{scope}
\end{tikzpicture}

\caption{\label{fig:reduction} Illustration of the parameterized reduction from an instance of {\sc$k$-Colored Clique} (left) to {\sc Consensus Patterns} using the string length as a parameter (right). An optimal solution $S=\sca ade$ and its alignment with each input string is given (positions producing a match are in bold). Note that vertices $\{a,c,e\}$ form a clique in $G$.
}
\end{figure}

We build an alphabet $\Sigma$ containing $V$ (i.e., one symbol per vertex) and two special characters $\sca$ and $\scb$.

Define string $\sv_i= \sca v_{1,i} v_{2,i}\ldots v_{k,i}$. Let $e=(v_{h,i},v_{h',i'})$ be the $j$th edge of $E$, $j\in [m]$. 
Define $\se_j$ as the string starting with $\sca$, followed by $k+1$ characters: all $\scb$, except for two positions: $\se_j[k+h+1]=v_{h,i}$ and $\se_j[h'+2]=v_{h',i'}$.

Let $N=m(k+2)+1$. The instance $\I$ of {\sc Consensus Patterns} contains $N$ occurrences of strings $\sv_i$, $i\in [n]$, and one occurrence of strings $\se_j$, $j\in[m]$. The target length is $\mathrm m=k+1$.

Note that due to the large value of $N$, any solution $S$ must have a minimal distance to the set of strings $\{\sv_i\mid i\in [n]\}$. Otherwise, (if it is, say, at the minimum distance plus one), the distance to the whole instance $\I$ increases by at least $N$, which cannot be compensated by the remaining strings $\se_j$ (which have size $m(k+2)<N$). Hence we first enumerate the optimal solutions for the set $\{\sv_i\mid i\in [n]\}$.

\begin{lemma}
The \emph{Consensus Patterns} of $\{\sv_i\mid i\in [n]\}$ (i.e., the strings of length $k+1$ at minimum  total distance from strings $\sv_i$) are the strings of the form $S=\sca v_{1, i_1} \ldots v_{k, i_k}$ with $i_1,\ldots, i_k\in [n]$. Such a string has a total distance of $(n-1)k$.
\end{lemma}
\begin{proof}
Since all strings in $\{\sv_i\mid i\in [n]\}$ have length $k+1$, any consensus pattern $S$ must be aligned with $\sv_i$ from the very first character. Hence $S$ is a consensus string of $\{\sv_i\mid i\in [n]\}$. The consensus strings of this set are obtained by taking the majority character at each position. Thus, $S[1]=\#$, and, for all $h\in[k]$, there exists $i_h$ such that $S[h+1]=\{v_{h,i_h}\}$.
\end{proof}

Consider now an optimal solution $S$ for $\I$. Let $\{i_h\mid h\in[k]\}$ be the set of indices as obtained from the lemma above. We show that the set of vertices $K=\{v_{h,i_h}\mid h\in[k]\}$ forms a clique of $G$ iff the distance is below a certain threshold. To this end, we compute the best possible alignment between $S$ and each string $\se_j$.

\begin{lemma}
Let $j\in [m]$. If both endpoints of edge $e_j$ are in $K$ then there exists an alignment of $S$ at distance $k-1$ from $\se_j$, otherwise the best possible alignment has distance $k$.
\end{lemma}
\begin{proof}
\newcommand{\s}{\mathrm s}
\renewcommand{\t}{\mathrm t}
Let $h_\s,h_\t,i_\s,i_\t$ be such that $e_j=(v_{h_\s,i_\s}, v_{h_\t,i_\t})$. 
There are two possible alignments of $S$ with $\se_j$: $S[1]$ is aligned either with $\se_j[1]$ or with  $\se_j[2]$. We compute the distance in both cases.

If $S[1]$ is aligned with $\se_j[1]$, then there is exactly one common character, namely $S[1]=\se_j[1]=\sca$. Indeed, for all $h\in[k]$, $S[h+1]\in V_h$ and $\se_j[h+1]\in V_{h-1}\cup\{x\}$, hence these two characters are different. The distance in this case is $k$.

If $S[1]$ is aligned with $\se_j[2]$, then first note that $S[1]=\sca \neq x= \se_j[2]$. 
Consider  index $h_{\s}$.  If $i_\s=i_{h_\s}$, then $S[{h_\s}]=v_{{h_\s},i_{h_\s}}=\se_j[{h_\s}+1]$, otherwise $S[{h_\s}] \neq \scb =\se_j[{h_\s}+1]$. 
Similarly for $h_\t$, $S[{h_\t}]=\se_j[{h_\t}+1]$ iff $i_\t=i_{h_\t}$. For other values of $h$ (i.e. $h\in [k]\setminus\{h_\s, h_\t\}$), $S[h]\neq \scb = \se_j[h+1]$. The distance is thus $k-1$ iff $i_\s=i_{h_\s}$ and $i_\t=i_{h_\t}$, it is at least $k$ otherwise.

Overall, if $i_\s=i_{h_\s}$ and $i_\t=i_{h_\t}$ the optimal alignment has distance $k-1$, otherwise the optimal alignment has distance $k$. 

\end{proof}

We can now conclude the proof. Let $S$ be an optimal solution of {\sc Consensus Pattern} for instance $\I$ and $K$ its corresponding set of vertices. The distance from $S$ to the $N$ copies of strings $\sv_i$ is $N(n-1)k$. The distance between $S$ and $\se_j$ is $k-1$ if both endpoints of $e_j$ are in $K$, and $k$ otherwise. $|E(K)|$ is the number of edges with both endpoints in $K$: the total distance from $S$ to strings $\se_j$ is thus $mk - |E(K)|$, and the total distance from $S$ to $\I$ is $N(n-1)k+mk-|E(K)|$. Overall, the optimal distance is at most $N(n-1)k+mk-\frac{k(k-1)}2$ if, and only if, $G$ contains a size-$k$ set of vertices $K$ with $|E(K)|\geq\frac{k(k-1)}2$, i.e. if $G$ contains a clique. 
\bibliographystyle{plain}
\bibliography{biblio}
\end{document}